\documentclass[11pt]{article}

\usepackage{amssymb}

\newsavebox{\theorembox}
\newsavebox{\lemmabox}
\newsavebox{\corollarybox}
\newsavebox{\propositionbox}
\newsavebox{\examplebox}
\newsavebox{\conjecturebox}
\newsavebox{\algbox}
\newsavebox{\qbox}
\newsavebox{\problembox}
\newsavebox{\definitionbox}
\newsavebox{\assumptionbox}
\newsavebox{\hypothesisbox}
\savebox{\theorembox}{\noindent\bf Theorem}
\savebox{\lemmabox}{\noindent\bf Lemma}
\savebox{\corollarybox}{\noindent\bf Corollary}
\savebox{\propositionbox}{\noindent\bf Proposition}
\savebox{\examplebox}{\noindent\bf Example}
\savebox{\conjecturebox}{\noindent\bf Conjecture}
\savebox{\algbox}{\noindent\bf Algorithm}
\savebox{\qbox}{\noindent\bf Question}
\savebox{\definitionbox}{\noindent\bf Definition}
\savebox{\problembox}{\noindent\bf Problem}
\savebox{\assumptionbox}{\noindent\bf Assumption}
\savebox{\hypothesisbox}{\noindent\bf Hypothesis}
\newtheorem{theorem}{\usebox{\theorembox}}
\newtheorem{lemma}[theorem]{\usebox{\lemmabox}}
\newtheorem{corollary}[theorem]{\usebox{\corollarybox}}
\newtheorem{proposition}[theorem]{\usebox{\propositionbox}}

\newenvironment{proof}{\noindent\bf{Proof.}\rm}{\hfill$\blacksquare$\bigskip}

\begin{document}

\title{Recoverable Values for Independent Sets}

\author{ \\Uriel Feige and Daniel Reichman
\thanks{Weizmann Institute of Science, PO Box 26, Rehovot 76100,
Israel. Email: {\tt{uriel.feige@weizmann.ac.il, daniel.reichman@gmail.com.}} Work supported in part by the Israel
Science Foundation (grant No. 873/08)}
}

\maketitle

\begin{abstract}
The notion of {\em recoverable value} was advocated in work of Feige, Immorlica, Mirrokni and Nazerzadeh [Approx 2009] as a measure of quality for approximation algorithms. There this concept was applied to facility location problems. In the current work we apply a similar framework to the maximum independent set problem (MIS). We say that an approximation algorithm has {\em recoverable value} $\rho$, if for every graph it recovers an independent set of size at least $\max_I \sum_{v\in I} \min[1,\rho/(d(v) + 1)]$, where $d(v)$ is the degree of vertex $v$, and $I$ ranges over all independent sets in $G$. Hence, in a sense, from every vertex $v$ in the maximum independent set the algorithm recovers a value of at least $\rho/(d_v + 1)$ towards the solution. This quality measure is most effective in graphs in which the maximum independent set is composed of low degree vertices. It easily follows from known results that some simple algorithms for MIS ensure $\rho \ge 1$. We design a new randomized algorithm for MIS that ensures an expected recoverable value of at least $\rho \ge 7/3$. In addition, we show that approximating MIS in graphs with a given $k$-coloring within a ratio larger than $2/k$ is unique games hard. This rules out a natural approach for obtaining $\rho \ge 2$.
\end{abstract}

\section{Introduction}

The notion of {\em recoverable value} was advocated in work of Feige, Immorlica, Mirrokni and Nazerzadeh~\cite{PASS} as a measure of quality for approximation algorithms. This notion leads to greater expressive power for stating the guarantees of approximation algorithms (compared to the standard notion of approximation ratio), by this leading to greater differentiation among the performance guarantees of different algorithms. The hope is that this concept will lead to the design of new algorithms with superior performance with respect to the recoverable value measure (regardless of whether the classic approximation ratio differs from that of existing algorithms), and moreover, that these algorithms will have
better performance in practice (at least in some interesting special cases).

In~\cite{PASS}, the term PASS approximation (where PASS is an acronym for {\em PArameterized by the Signature of the Solution}) was used in order to capture the two main features that we wish the recoverable value to have. One feature is that the recoverable value is expressed in terms of properties of the (unknown) solution, rather than of the input instance. The other is that the property of the solution that it refers to is not some aggregate property (such as average degree), but rather some {\em signature} in which the contribution of each individual component of the solution is considered separately. Rather than try to present general principles here, let us focus on the problem studied in the current paper, that of maximum independent set, and specialize the notion of recoverable value to this problem.

We use the following notation.
All graphs in this work are undirected. The degree $d(v)$ of vertex $v$ is the number of neighbors of $v$. The set of neighbors of $v$ is $N(v)$. The average degree of a graph is denoted by $d_{avg}$, and $d_{avg(U)} = \frac{1}{|U|}\sum_{v\in U} d(v)$ denotes the average over degrees of vertices in a set $U$. An independent set in a graph is a set of vertices $I$ such that every two vertices in $I$ are non-neighbors. We shall refer to the problem of finding an independent set of maximum cardinality as MIS.  The independence ratio of $G=(V,E)$ is $\alpha(G) = |I_{max}|/|V|$, where $|I_{max}|$ is the size of a maximum independent set $I$. In the maximum weight independent set problem MWIS, every vertex $v$ has a nonnegative weight $w_v$ and the goal is to find an independent set of maximum weight.

We now present our notion of recoverable value for MIS. For an independent set $I$, we define its {\em signature} to be the sequence of degrees of vertices in $I$. The value that we shall want to recover from each vertex of $I$ depends on its degree. With every degree $d$ we associate a recoverable value of $0 \le \rho_d \le 1$, and wish our approximation algorithms to find an independent set of size at least $\sum_{v \in I} \rho_{d(v)}$. The independent set $I$ is not known to the algorithm -- it is only used in analyzing its performance measure. Hence this performance guarantee holds simultaneously with respect to all independent sets in the graph, including the independent set $I$ that happens to maximizes the expression $\sum_{v \in I} \rho_{d(v)}$ (this $I$ might not be the maximum independent set in the graph). The notion of recoverable value easily generalizes to MWIS, by multiplying each term by $w(v)$. For randomized algorithms, one considers the expected size (or weight for MWIS) of the independent set that they return.

Intuitively, the smaller the degree of $v$ the more likely algorithms are to place $v$ in an independent set (because $v$ excludes a smaller number of other vertices), and hence the higher we would like its recoverable value to be. Hence it is natural for $\rho_d$ to be a non-increasing function of $d$. We find it convenient to introduce a parameter $\rho$ (that we shall attempt to maximize later) and to consider the function $\rho_d = \min[1,\rho/(d+1)]$. (Not allowing the recoverable value to exceed~1 is necessary, as we cannot find a solution larger than the optimal solution.) For graphs of minimum degree at least $\rho - 1$, such approximation algorithms need to find an independent set of size at least $\sum_{v \in I} \rho/(d(v) + 1)$.

We refer to $\rho$ in the expression $\min[1,\rho/(d(v)+1)]$ as the {\em canonical} recoverable value (though we sometimes omit the word canonical for brevity).  For comparison with some previously published algorithms, it is useful to note that for every set $U$ of vertices, $\sum_{v\in U} 1/(d_v + 1) \ge |U|/(d_{avg(U)}+1)$ (with equality only if all vertices in $U$ have the same degree). Hence for a given value of $\rho$, our notion of recoverable value that is based on the signature of $I$ (its degree sequence) is more demanding than had we only considered the average degree of vertices in $I$.

\subsection{Our results}

It is not hard to see that some known algorithms achieve a canonical recoverable value $\rho \ge 1$. We show (see Proposition~\ref{pro:negative} in the appendix) that they do not achieve a value of $\rho$ bounded away from~1.

One can readily observe that $\rho_0 = \rho_1 = 1$. We use a procedure that we call {\em 2-elimination} to that that for MIS one may enforce $\rho_2 = 1$ as well.

Thereafter, we design relatively simple new algorithms that achieve a recoverable value of $\rho \ge 2$ for MWIS and $\rho \ge 7/3$ for MIS. It is NP-hard to achieve $\rho = 4$, as this will imply exact solution of MIS in 3-regular graphs, a problem that is NP-hard and APX-hard~\cite{Berman}. However, in graphs in which the minimum degree $\delta$ is large we provide an algorithm achieving $\rho \ge \Omega(\log \delta / \log\log \delta)$ for MWIS.

Our investigations of the best recoverable value achievable by our approaches also lead us to consider MIS in $k$-colored graphs (graphs for which a $k$-coloring is given). For this problem a $2/k$ approximation ratio is known~\cite{Hochbaum}, and we show that improving it would refute the unique games conjecture.

Our proofs in some cases provide simple alternatives to previously published related results.

\subsection{Related work}
\label{sec:relatedwork}

The strong inapproximability results for MIS~\cite{Hastad} justify searching other measures of performance guarantees, and the notion of recoverable value is one such candidate. It considers degrees of vertices, and hence it is instructive to recall known approximation algorithms for MIS and MWIS and how their performance depends on the degree sequence of the graph.

{\bf Random permutation.} Choosing a permutation uniformly at random and taking all vertices that appear prior to their neighbors in the order induced by the permutation produces an independent set whose expected weight is at least $\sum_{v \in V}{\frac{w_v}{d(v)+1}}$ (see \cite{AS}, for example). This guarantees a recoverable value of $\rho \ge 1$.

{\bf Greedy.} For MIS, iteratively picking a minimum degree vertex, adding it to an independent set $I$ and deleting the vertex and its neighbors from the graph is guaranteed to find an independent set of size at least $\sum_{v \in V}{\frac{1}{d(v)+1}}$ \cite{Gr,Wei}. Halldorsson and Radhakrishnan \cite{HR97} showed that this greedy algorithm produces an independent set of size at least $|V|\frac{1 + \alpha^2}{d_{avg} + 1 + \alpha}$ (where $\alpha$ denotes the fraction of vertices in the maximum independent set). For MWIS, {\em weighted greedy} that iteratively picks a vertex $v$ with minimum $w_v/(d_v + 1)$ is guaranteed to find an independent set of size at least $\sum_{v \in V}{\frac{w_v}{d(v)+1}}$~\cite{Sakai}.

{\bf LP.} An integer programming formulation of MWIS is:

maximize $\sum_{i \in V}w_i x_i$

subject to

$x_i + x_j \le 1$ for every edge $(i,j)$.

$x_i \in \{0,1\}$ for every vertex $i$.

Consider the LP relaxation of this program where each $x_i \in [0,1]$. A well known result due to Nemhauser and Trotter \cite{NT} asserts that there is an optimal solution for the relaxation such that for every $i,x_i \in \{0,\frac{1}{2},1\}$. Moreover, such an optimal solution can be found in polynomial time.

{\bf LP+greedy.} Consider the following algorithm. Find an optimal half integral solution to the LP, discard all the vertices assigned~0, keep all the vertices assigned~1, and run the greedy algorithm on the graph induced by all vertices that are assigned $1/2$. This algorithm was analyzed for connected graphs. Hochbaum \cite{Hochbaum} proved an approximation ratio of $\frac{2}{d_{avg}+1}$, and Halldorsson and Radhakrishnan \cite{HR97} (based on their improved analysis of the greedy algorithm) proved an approximation ratio of $\frac{5}{2d_{avg}+3}$.

{\bf SDP.}
Using semidefinite programming Halldorsson \cite{H} provided approximation ratios of $\Omega(\frac{\log d_{avg}}{d_{avg}\log \log d_{avg}})$ for MIS, and $\Omega(\frac{\log \delta}{\delta \log \log \delta})$ for MWIS on {\em $\delta$-inductive} (a.k.a. {\em $\delta$-degenerate}) graphs (in some permutation over the vertices $\delta$ is the maximum backward degree).

{\bf Local Search.}
In graphs of degree bounded by $\Delta$, algorithms based on local search were shown to achieve an approximation ratio of roughly $5/(\Delta+3)$, with some distinction between the cases of odd and even $\Delta$ (see~\cite{Berman,Chlebik}). We remark that in the special case of $\Delta$-regular graphs, the notion of recoverable value becomes equivalent to the traditional notion of approximation ratio, and our results are not as strong in this case as those achieved by local search. On the other hand, our algorithms are much faster than the local search algorithms.

In terms of hardness results, Austrin, Khot and Safra \cite{AKS} proved that approximating independent set in graphs of maximum degree $\Delta$ within a ratio larger than $\frac{(\log \Delta)^2}{\Delta}$ is unique games hard. Recall that we present hardness results for finding independent sets in graphs where a $k$-coloring is given. They essentially match the $2/k$ bounds achieved by known approximation algorithms~\cite{Hochbaum}. We are not aware of previous published hardness results for this problem, but there are some results for related problems on hypergraphs~\cite{GS}, and hardness results for MIS in graphs with bounded chromatic number but when no coloring is given~\cite{GSinop}.

\subsection{Our techniques}
\label{sec:techniques}

Vertices of very low degree receive special treatment. For MIS we show that a certain {\em 2-elimination} technique implies $\rho_2 = 1$. (See Section~\ref{sec:elimination} for subtleties involved in the statement of this result.)

Thereafter, following a recipe suggested in~\cite{PASS} (though there the problem considered was different, facility location), we present an algorithm based on a so called recoverable value LP (see Theorem~\ref{thm:MWIS}). This gives recoverable value of $\rho = 2$ for MWIS. For MIS we improve over this bound by using a new combination of some of the algorithms presented in Section~\ref{sec:relatedwork}. (A new combination is indeed needed, as we also provide examples showing that plain use of these algorithms does not even achieve $\rho$ bounded away from~1.)

Given a graph $G$, choose a random permutation $\pi$ on the vertices. We say that vertex $v$ is in {\em layer} $L_i$ with respect to $\pi$ if exactly $i-1$ of $v$'s neighbors appear before $v$ in $\pi$. For $k \ge 1$, let $G_k$ be the subgraph of $G$ induced on the vertices of the first $i$ layers. The random permutation algorithm referred to in Section~\ref{sec:relatedwork} simply returns $G_1$ which is an independent set. Our new algorithms will instead consider $G_k$ with small value of $k$ (depending on the context, we shall take $k \in \{2,3,\delta + 1\}$,  where $\delta$ denotes the minimum degree), and on $G_k$ run some algorithm from Section~\ref{sec:relatedwork}. The advantage of our approach (in the context of recoverable value) is that low degree vertices are more likely to end up in $G_k$. Moreover, $G_k$ (for small values of $k$) has special structure that makes finding large independent sets easier. For example, it turns out that $G_k$ is $k$-colorable. This suffices for obtaining a recoverable value of $\rho = 2$ (for MWIS on the original $G$), but does not guarantee anything better (by our new hardness of approximation results for approximating MIS in $k$-colored graphs). The fact that $G_k$ is $(k-1)$-inductive allows us to obtain a recoverable value of $\rho = \Omega(\log \delta/ \log\log \delta)$ (which improves over $\rho = 2$ when $\delta$ is large). Our unconditional improvement over $\rho = 2$ uses $G_3$ and applies only to MIS. For $G_3$ we consider its average degree, which is no longer a deterministic property but rather a function of several random variables (number of vertices in each layer). After showing that the $\frac{5}{2d_{avg}+3}$ approximation ratio of~\cite{HR97} extends to disconnected graphs (see Theorem~\ref{thm:isolated}), we show that the random variables can safely be replaced by their expectations (which are deterministic quantities), by this considerably simplifying the analysis. It was previously known that working with expectations often simplifies the analysis of randomized algorithms (through the use of linearity of the expectation), but the reason why it applies in our context is more delicate than usual, and may have applications also elsewhere. This leads to a canonical recoverable value $\rho \ge 15/7$. Thereafter we improve to $\rho \ge 7/3$ by designing a new approximation algorithm for MIS in graphs of small average degree.

\section{Handling low degree vertices}
\label{sec:elimination}

It is instructive to consider separately $\rho_d$ for $0 \le d \le 3$.

For isolated vertices, $\rho_0 = 1$ both for MIS and for MWIS, as they can be included in the optimal solution without changing the degree of any other vertex. Note an important point here. Our notion of  recoverable value treats every vertex of $I$ separately. Hence we can remove isolated vertices without effecting the remaining vertices in $I$. This might not have been so simple had we consider aggregate properties such as the average degree of vertices in $I$. Removing an isolated vertex increases the average degree for the remaining vertices.

Likewise, $\rho_1 = 1$. For MIS this follows because every vertex of degree~1 can be included in the independent set, and its neighbor removed. At most one of these two vertices is in any independent set. For other vertices in $I$, this process can only lower their degrees, and hence their recoverable value does not decrease (since $\rho_d$ is non-increasing). For MWIS the argument is slightly more delicate. Let $u$ be a vertex of degree~1 and let $v$ be its neighbor. If $w_u \ge w_v$ then the same argument as in the unweighted case works. However, if $w_u < w_v$ one should remove $u$ from $G$, and reduce the weight of $v$ to $w_v - w_u$, thus obtaining a new graph $G'$. Any independent set $I'$ in $G'$ can be extended to an independent set $I$ in $G$ whose value is higher by $w_v$. If $v \in I'$, simply increase the weight of $v$ back to its original value. If $v \not\in I'$, then include $u$ in $I$. As the transformation from $G$ to $G'$ only removed $u$ and did not increase the degree of any remaining vertex, monotonicity of $\rho_d$ implies that the recoverable of degree~1 vertices is~1.

The above arguments also imply that we can remove recursively vertices of degree at most~1 with no harm to the recoverable value. Hence we may assume that all graphs have minimum degree at least~2.

We show that for MIS essentially $\rho_2 = 1$ (see Theorem~\ref{thm:smalldegree} for an exact statement). Consider a graph $G$ of minimum degree~2. Let $u$ be a vertex of degree~2, and let $v$ and $w$ be its neighbors. We now describe a process that we call {\em 2-elimination}. 
If there is an edge $(v,w)$ we can safely add $u$ to the independent set (and remove vertices $v$ and $w$) because at most one of $u,v,w$ is in any independent set, and $u$ can replace any of $v$ and $w$ in an independent set. Hence it remains to deal with the case that $v$ and $w$ are not neighbors of each other. In this case, remove $u$ from $G$, merge $v$ and $w$ to become a new vertex $u'$ whose neighbors are the original neighbors of $v$ and $w$ (except for $u$ that was removed from the graph). This gives a new graph $G'$. Any independent set $I'$ in $G'$ can be extended to an independent set $I$ in $G$ whose size is larger by~1. If $u' \in I'$, replace it by $v$ and $w$. If $u' \not\in I'$, then include $u$ in $I$.

Now let us analyze how the recoverable value changes when transforming from $G$ to $G'$. We assume here canonical recoverable values. Without loss of generality, either $u \in I$ or both $v$ and $w$ are in $I$. In $G'$, we take an independent set $I'$ that is induced by $I$ in a natural way (if $u\in I$ than $u$ is simply lost, if $v,w \in I$ then $u' \in I'$). Note that $|I| = |I'| - 1$, hence we would like to show that their recoverable values differ by at most~1. If $u \in I$ then $I'$ differs from $I$ by the recoverable value of $u$ which is indeed at most~1. If $v,w \in I$, then the recoverable values of $G$ and $G'$ differ by

\begin{equation}
\label{eq:removed2}
\min[1,\frac{\rho}{d_v+1}] + \min[1,\frac{\rho}{d_w+1}] - \min[1,\frac{\rho}{d_{u'}+1}]
\end{equation}

Using the facts that $d_v,d_w \ge 2$ and $d_{u'} \le d_v + d_w - 2$ the value of (\ref{eq:removed2}) is at most~1 whenever $\rho \le 10/3$ (the worst choice of parameters being $d_v = d_w = 3$ and $d_{u'} = 4$). Hence if one considers canonical recoverable values, then if $\rho \ge 3$ this implies (by definition) that $\rho_2 = 1$, and if $\rho \le 10/3$ then 2-elimination achieves $\rho_2 = 1$ (without affecting $\rho$ for vertices of degrees larger than~2).

The above discussion establishes:

\begin{theorem}
\label{thm:smalldegree}
When using canonical recoverable values for MIS then regardless of the value of $\rho$ one may assume that $\rho_0 = \rho_1 = \rho_2 = 1$.
\end{theorem}

Theorem~\ref{thm:smalldegree} cannot be extended to $\rho_3$. Moreover, there is some fixed $\epsilon > 0$ such that $\rho_3 \le 1 - \epsilon$. This follows from the fact that approximating MIS in 3-regular graphs in APX-hard~\cite{Berman}.

\section{Algorithms for MWIS}
\label{sec:MWIS}

In this section we assume that $G=(V,E)$ has been preprocessed to include all isolated vertices in the output independent set. This allows us to simplify notation from $\min[1,\rho/(d(v) + 1)]$ to $\rho/(d(v) + 1)$.

\begin{theorem}
\label{thm:MWIS} Let $G=(V,E)$ be a weighted graph without isolated vertices. There is a polynomial time algorithm for MWIS achieving a recoverable value of $\rho = 2$. Namely, the output of the algorithm is an independent set of weight at least $\sum_{v \in I}\frac{2 w_v}{d(v)+1}$.
\end{theorem}

\begin{proof}
We present two different polynomial time algorithms that achieve the desired bounds. One is based on linear programming. The other is much faster, but randomized.

\noindent {\bf LP algorithm.} Consider the following recoverable value LP (the RV LP). $x_i$ is
a variable that indicates whether vertex $i$ is in the independent
set.

maximize $\sum_{i \in V} \frac{w_i}{d(i) + 1} x_i$

subject to

$x_i + x_j \le 1$ for every edge $(i,j)$.

$0 \le x_i \le 1$ for every vertex $i$.

The indicator vector of any independent set $I$ is a feasible solution
to the recoverable value LP. Moreover, the value of the LP then would be the recoverable value with respect to $I$ (up to a scaling factor of $\rho$). Hence the optimal value of the LP is
at least the desired recoverable value with respect to best independent set $I$ (scaled by $1/\rho$). Treating $\frac{w_i}{d(i) + 1}$ as a weight of vertex $i$, the results of~\cite{NT} imply that this LP has a half-integral
optimal solution, and moreover that such a solution can be found in polynomial time.

Include in the independent set all vertices with $x_i = 1$ (getting credit $w_i$ which is at least twice the credit $w_i/(d(i) + 1)$ that the LP got for them), and discard all
vertices of $x_i = 0$ (the LP got no credit for them).
Let $G_{1/2}$ be the weighted graph induced on all vertices assigned $1/2$. Running weighted greedy on this graph ensures a solution of value $\sum_{i\in V_{1/2}} w_i/(d_i + 1)$, whereas the LP (by having $x_i = 1/2$) got only half this credit. Hence after the rounding we obtain an integral solution of value at least twice that of the RV LP, implying $\rho \ge 2$.

\noindent {\bf A fast randomized algorithm.}
Choose a random permutation and consider $G_2$ as in Section~\ref{sec:techniques}. This graph is a forest (can be verified by orienting edges towards earlier vertices in the permutation). As every vertex $v$ within $I$ belongs to $G_2$ with probability $\frac{2}{d(v)+1}$, the expected weight of an independent set within $G_2$ is at least $\sum_{v \in I}\frac{2w_v}{d(v)+1}$. MWIS can be found in forests in linear time, proving Theorem~\ref{thm:MWIS}.
\end{proof}

The bounds $\sum_{v \in V}\frac{w_v}{d(v)+1}$ (achieved by the random permutation algorithm of Section~\ref{sec:relatedwork}) and $\max_{I} \sum_{v\in I} 2w_v/(d(v)+1)$ of Theorem~\ref{thm:MWIS} are incomparable. Nevertheless, it is not hard to see that both algorithms of Theorem~\ref{thm:MWIS} ensure not only the bound $\max_{I} \sum_{v\in I} 2w_v/(d(v)+1)$ claimed in the statement of the theorem, but also the bound $\sum_{v \in V}\frac{w_v}{d(v)+1}$. For the LP algorithm this is attained by the feasible solution that assigns $1/2$ to all variables. For the fast randomized algorithm this follows because returning the first layer is a legitimate output for it.

We now improve the recoverable value in the special case when the minimum degree $\delta$ in the input graph is sufficiently high.

\begin{theorem}
\label{thm:delta} Let $G=(V,E)$ be a weighted graph with minimum degree $\delta$. There is a polynomial time algorithm for MWIS achieving a recoverable value of $\rho = \Omega(\log \delta/ \log\log \delta)$.
\end{theorem}

\begin{proof}
Choose a random permutation and consider $G_{\delta+1}$ as defined in Section~\ref{sec:techniques}. This graph is a $\delta$-inductive (orienting edges towards earlier vertices in the permutation, no vertex has outdegree larger than $\delta$). As every vertex $v$ within $I$ belongs to $G_{\delta + 1}$ with probability $\frac{\delta + 1}{d(v)+1}$ (here we used the assumption that $d(v) \ge \delta$), the expected weight of an independent set within $G_{\delta + 1}$ is at least $\sum_{v \in I}\frac{w_v(\delta + 1)}{d(v)+1}$. Run the SDP algorithm from~\cite{H} with approximation ratio $\Omega(\frac{\log \delta}{\delta \log \log \delta})$ (on $\delta$-inductive graphs) to obtain a recoverable value with $\rho = \Omega(\log \delta/ \log\log \delta)$.
\end{proof}

\section{Algorithms for MIS}
\label{sec:MIS}

Here we show that for unweighted graphs there are randomized algorithms achieving a recoverable value of $\rho$ strictly larger
than~2. (The value of $\rho$ will be further improved in Section~\ref{sec:improve}, by extending the principles developed in the current section.) As in Section~\ref{thm:MWIS}, we wish to simplify notation from $\min[1,\rho/(d(v) + 1)]$ to $\rho/(d(v) + 1)$.
This requires that the input graph $G=(V,E)$ has no vertices of degree less than~2. This can be assumed without loss of generality (see Section~\ref{sec:elimination}).

The basic idea of our algorithm is as follows. Pick a random permutation over the vertices and consider the graph $G_3$ induced on layers $L_1 \cup L_2 \cup L_3$. The expected size of $I\cap G_3$ is now $\sum_{v \in I}\frac{3}{d(v)+1}$ (Here we use the fact that the minimum degree in $G$ is 2). One would expect all three layers $L_1$,$L_2$ and $L_3$ to be of equal size. Moreover, every vertex in $L_1$ contributes no edge to $G_3$, every vertex in $L_2$ contributes at most one edge to $G_3$, and every vertex in $L_3$ contributes at most two edges to $G_3$.
Hence one would expect the average degree of $G_3$ to be at most $2$. Recall that the algorithm of \cite{HR97} (LP+greedy) obtains an approximation ratio of $\frac{5}{2d_{avg} + 3}$ (on connected graphs). Hence, applying this algorithm to $G_3$ (and using the bounds of~\cite{HR97} even though $G_3$ need not be connected), one may hope to attain recoverable value of $\frac{15}{7}\sum_{v \in I}\frac{1}{d(v)+1}$.

Summarizing, we have algorithm PLG (Permute, LP, Greedy).

\begin{enumerate}

\item
Pick a random permutation over the vertices and consider the graph $G_3$ induced on layers $L_1 \cup L_2 \cup L_3$.

\item
Find a half-integral solution to the LP for MIS on $G_3$. Put the vertices with value~1 in the independent set and remove the vertices with value~0.

\item Run the greedy algorithm on the subgraph induced on the vertices that remain from $G_3$.

\end{enumerate}

The above argument of why PLG achieves a value of $\rho = 15/7$ had two gaps in it. One relates to the assumption that $G_3$ is connected (which might not hold), and the other to the assumption that sizes of layers are equal to their expectations. We shall deal with each one of them separately.

As noted, the approximation ratio of $\frac{5}{2d_{avg} + 3}$ of the algorithm of~\cite{HR97} assumes that the graph is connected. Indeed, trying to use this expression with $d_{avg} < 1$ (which may well be the case for graphs with isolated vertices), one obtains approximation ratios better than~1 which of course cannot be true. Nevertheless, when the average degree is at least~2, the following theorem shows that the bound of $\frac{5}{2d_{avg} + 3}$ does apply, regardless of whether the graph is connected or not. Observe that all connected graphs are captured by the theorem (either they are trees and then greedy solves them optimally, or their average degree is at least~2), and hence our proof can also replace the one given in~\cite{HR97} for connected graphs.

\begin{theorem}
\label{thm:isolated} If $G$ is a graph of average degree at least~$2$ then MIS can be approximated within ratio of $\frac{5}{2d_{avg} + 3}$.
\end{theorem}

\begin{proof}
Recall that by the results of~\cite{HR97}, the approximation ratio of the greedy algorithm on graphs with average degree $d_{avg}$ and independence ratio $\alpha$ is $f(\alpha)=\frac{1+\alpha^2}{(d_{avg}+1+\alpha)\alpha}$. This ratio as a function of $\alpha$ is decreasing in the range $(0,1/2]$. For $\alpha = 1/2$ we obtain the desired approximation ratio of $\frac{5}{2d_{avg} + 3}$. Hence to prove Theorem~\ref{thm:isolated} it suffices to show that we can assume that $\alpha(G) \le 1/2$.

Consider an optimal half integral solution to the standard LP relaxation of the MIS problem (as implied by~\cite{NT}). Let $ONE$ denote the set of variables receiving 1, $ZERO$ the set of variables receiving 0, and $HALF$ the set of variables receiving $1/2$. Let $H$ be the subgraph induced on the vertices whose corresponding variables are in $HALF$. The independence ratio of $H$ is at most $1/2$, as desired. Necessarily $|ONE| \ge |ZERO|$ (otherwise they would both be in $HALF$). Remove $ONE$ and $ZERO$ and all edges connected to them (and thus only $H$ is left), and add instead $|ONE|$ isolated edges (one edge for each vertex of $ONE$), thus obtaining a new graph $G'$. Observe that the size of the maximum independent set does not change, but $\alpha(G') \le 1/2$ as desired. It remains to analyze the average degree of $G'$. In the removal phase $|ONE| + |ZERO|$ vertices and at least $|ZERO|$ edges are removed (every vertex in $ZERO$ has an edge to $ONE$, otherwise it could be moved to $HALF$). Thereafter $2|ONE|$ vertices and $|ONE|$ edges are added. Hence altogether exactly $|ONE| - |ZERO|$ vertices and at most $|ONE| - |ZERO|$ edges are added. If the average degree of $G$ is at least~2, the average degree of $G'$ cannot be larger than that of $G$.

Summarizing, we have shown how to transform $G$ into a new graph $G'$ for which $\alpha(G') \le 1/2$, the average degree of $G'$ is at most that of $G$, and the maximum independent sets in $G$ and $G'$ have the same size. Applying greedy to $G'$ is the same as taking $ONE$ into the solution and applying greedy only on $H$, which is precisely the algorithm of~\cite{HR97}. By the properties of $G'$, the approximation ratio is at least as desired by Theorem~\ref{thm:isolated}.
\end{proof}

Apparently, the result of Theorem~\ref{thm:isolated} can be extended also
to $d_{avg} \ge 3/2$ (Halldorsson, private communication), though this is
not needed in our paper.

Now let us deal with the issue of expectations. Let us first explain the problem. We have various information about expectations. Namely, $E[|I \cap V(G_3)|] = \sum_{v\in I} 3/(d(v) + 1)$, and $E[|L_1|] = E[|L_2|] = E[|L_3|]$. Moreover, the number of edges in $G_3$ is at most $2|L_3| + |L_1|$. If things behave exactly as expectation the average degree of $G_3$ is at most~2, Theorem~\ref{thm:isolated} applies and we get an approximation ratio of $\frac{5}{7}\sum_{v\in I} 3/(d(v) + 1)$ as desired. However, there is also variability in the above random variables, and it can be quite large. (The complete bipartite graph $K_{3,d}$ illustrates this variability.) With some probability the average degree in $G_3$ may be larger than~2, and with some probability smaller. Hence the approximation ratio on $G_3$ is a random variable. Moreover, when the average degree of $G_3$ is too small, the bounds of Theorem~\ref{thm:isolated} no longer hold. Moreover, the size of the maximum independent set in $G_3$ might be correlated with the average degree in complicated ways.

We present here a very simple way to handle all the above complications.

\begin{theorem}
\label{thm:expectation}
The expected size the independent set found by algorithm PLG is at least $\frac{15}{7}\sum_{v \in I}\frac{1}{d(v)+1}$, where $I$ is any independent set in $G$.
\end{theorem}

\begin{proof}
Our proof uses two important facts. One is that Theorem~\ref{thm:isolated} applies also to disconnected graphs. The other (which the reader may verify) is that algorithm PLG when run on a disconnected graph gives the same outcome (or more formally, the same probability distribution on outcomes) as that when PLG is run on each connected component separately.

Let $\epsilon > 0$ be any desired level of accuracy with which we want to estimate the performance guarantee of PLG. Let $G$ be an arbitrary graph on which we run PLG, let $n$ be the number of its vertices, and let $I$ be an independent set in $G$. As a thought experiment, make $N$ disjoint copies of $G$, where $N$ is chosen to be sufficiently large as a function of $\epsilon$ and $n$. Call the resulting graph on $nN$ vertices $G_N$, and let $I_N$ be the independent set composed from the $N$ copies of $I$. Run PLG on $G_N$. With respect to PLG, the random variables $L_1$, $L_2$, $L_3$, and $I_N \cap G_3$ are all concentrated around their expectation within relative errors that are $o(\epsilon)$ (by our large choice of $N$ and the fact that these random variables are each a sum of $N$ bounded and independent random variables, one for each copy of $G$). Hence on $G_N$, with probability that can be made $(1 - \epsilon/2)$ PLG finds an independent set of size at least $(1 - \epsilon/2)\frac{15}{7}\sum_{v \in I_N}\frac{1}{d(v)+1}$. Hence in expectation, per copy of $G$, the size of the independent set found is at least $(1 - \epsilon)\frac{15}{7}\sum_{v \in I_N}\frac{1}{d(v)+1}$. Letting $\epsilon$ tend to~0 Theorem~\ref{thm:expectation} is proved.
\end{proof}

\section{Improved recoverable values for MIS}
\label{sec:improve}

The main point of Theorem~\ref{thm:expectation} is that algorithm PLG has a recoverable value with $\rho$ strictly larger than~2. This value is at least $15/7$ (as the theorem shows) and at most~3 (for example, on the complete bipartite graph $K_{3,n-3}$), and it would be interesting to narrow this gap. In this section we use the notion of 2-elimination from Section~\ref{sec:elimination} do design variations on algorithm PLG for which we prove a value of $\rho$ strictly larger than $15/7$.

Section~\ref{sec:elimination} shows that for $\rho \le 10/3$, by repeatedly applying 2-elimination as long as possible, we may assume that $G$ has minimum degree~3. This allows us to replace $G_3$ by $G_4$ in algorithm PLG, capturing a higher fraction of vertices of $I$. (In contrast, if there are vertices of degree~2 in $I$, their fraction in $G_3$ is already~1 and cannot increase in $G_4$.) This gain is offset to some extent by the higher average degree of $G_4$ compared to $G_3$, which gives a poorer guarantee in the approximation ratio $\frac{5}{2d_{avg}+3}$ of Theorem~\ref{thm:isolated}. Still, we get a recoverable value of $\rho \ge 4\frac{5}{2 \cdot 3 + 3} = 20/9$. More specifically, for the original $I$ in $G$, vertices of degree at most~2 contribute~1 to the recoverable value, and vertices of degree at least~3 contribute at least $20/9(d_v + 1)$.

\begin{corollary}
\label{cor:20/9}
A canonical recoverable value $\rho = 20/9$ is achievable for MIS.
\end{corollary}

We can further improve the recoverable value by improving over the performance guarantees of Theorem~\ref{thm:isolated} (which are based on those of~\cite{HR97}) for graphs of a given average degree. We use 2-elimination to do so for the special case of $d_{avg} = 2$, which captures the graph $G_3$.

\begin{theorem}
\label{thm:7/9}
On graphs with average degree~2 (or less), MIS can be approximated within a ratio of $7/9$.
\end{theorem}

\begin{proof}
Let $G=(V,E)$ be a graph of average degree $2$. We shall have a sequence of partitions of $V$ into $V_1$ and $V_2$, and consider the graphs $H_1$ induced on $V_1$ and $H_2$ induced on $V_2$. 
Initially,
$V_1=V$ and $V_2=\emptyset$. Move vertices from $V_1$ to $V_2$ by iteratively using any of the following {\em simplifications} (as long as they are applicable to the current graph $H_1$).

{\bf 0-elimination}: move isolated vertices from $V_1$ to $V_2$.

{\bf 1-elimination}: move vertices $u$ of degree $1$ and their neighbors from $V_1$ to $V_2$. That is, if $(u,v)$ is an edge in $H_1$ and $u$ has degree~1 in $H_1$, move both $u$ and $v$ to $V_2$ (and hence the edge $(u,v)$ to $H_2$).

{\bf 2-elimination}: Consider a vertex $v_1$ of degree $2$ in $H_1$, incident to  $v_2,v_3 \in V_1$. If $v_2,v_3$ are neighbors, move all three vertices  to $V_2$ (inducing a triangle in $H_2$). Otherwise move $v_1$ and $v_2$ to $V_2$ ($v_2$ is chosen arbitrarily, $v_3$ would do just as well), leave $v_3$ in $V_1$, but change $G$ (and hence also the induced graphs $H_1$ and $H_2$) by connecting $v_3$ to those neighbors of $v_2$ that are still in $V_1$, and disconnecting $v_2$ from these neighbors. Observed that the edge $(v_1,v_2)$ is now in $H_2$.

{\bf NT} (for Nemhauser-Trotter): Run the LP for MIS on $H_1$ and find a half integral optimal solution (as in~\cite{NT}). If there are vertices with integer values (0 or 1), move them from $V_1$ to $V_2$.

Before continuing to describe our algorithm, let us present some consequences of the simplifications.

\begin{lemma}
\label{lem:H2}
$H_2$ has a maximum independent set $I_2$ such that there is no edge between $I_2$ and $V_1$. Moreover, such an independent set can be found in polynomial time.
\end{lemma}

\begin{proof}
The proof is by induction on the number of simplification steps. The base case is when $H_2$ is the empty graph, and then the lemma trivially holds. Given that the lemma holds for $H_2$, let $H_2'$ be the graph after one additional simplification step. The inductive hypothesis implies that in $H_2$ one finds a maximum independent set with no neighbors in $V_1$, and hence no edges to those vertices newly added to $H_2'$. Hence it remains to show that in the subgraph induced on those vertices newly added to $H_2'$ one can find a maximum independent set not connected to any vertex of the new $V_1$. For a 0-elimination, take the new vertex, for a 1-elimination take the vertex $u$ of degree~1, for a 2-elimination in which a triangle was moved take $v_1$, in a 2-elimination in which an edge $(v_1,v_2)$ was moved (and other neighbors of $v_2$ were transferred to $v_3$) take $v_2$, and for NT take those vertices of value~1. (These 1-vertices cannot have neighbors in the new $V_2$ since vertices there had fractional value~1/2. Moreover, the 1-vertices form a maximum independent set on the subgraph induced on vertices of value~0 and~1. Otherwise, there is a set of vertices of value~0 with fewer neighbors of value~1, but then the value of the LP could be raised by changing all these values to 1/2.)
\end{proof}

Let $G'$ be the graph obtained from $G$ by using the above simplifications. Observe that $G'$ might not be equal to $G$ because 2-elimination moves endpoints of edges from one vertex to another. Nevertheless:

\begin{lemma}
\label{lem:transform}
Given any independent set $I'$ in $G'$ one can find in polynomial time an independent set $I$ of equal size in $G$, and vice versa.
\end{lemma}

\begin{proof}
The only simplification step that makes $G'$ different from $G$ is a $2$-elimination step in which the neighbors $v_2,v_3$ of $v_1$ do not share an edge. Consider $G$ with the partition $V_1,V_2$ just before the $2$-elimination occurs and $G'$ with the partition $V_1', V_2'$ immediately after the elimination took place. Let $I$ be an independent set in $G$. If $I$ contains no vertex from $v_1,v_2, v_3$, then $I$ is independent also in $G'$. If $I$ contains one vertex from $v_1,v_2, v_3$, put $v_2$ in $I'$. Lemma ~\ref{lem:H2} implies that we may assume that $I \cap V_2$ has no edges to $v_2$, and hence $I'$ is indeed an independent set. If $I$ contains two vertices from $v_1,v_2, v_3$, then they must be $v_2$ and $v_3$, and then $I$ is also an independent set in $G'$. Observe that $I$ cannot contain all three $v_1,v_2, v_3$. Conversely, let $I'$ be an independent set in $G'$. If $I'$ contains no vertex from $v_1,v_2, v_3$, then $I'$ is independent also in $G$. If $I'$ contains one vertex from $v_1,v_2, v_3$, put $v_1$ in $I$. Lemma ~\ref{lem:H2} implies that we may assume that $I' \cap V_2$ has no edges to $v_1$, and hence $I$ is indeed an independent set. If $I'$ contains two vertices from $v_1,v_2, v_3$, then they must be $v_2$ and $v_3$, and then $I'$ is also an independent set in $G$. Observe that $I'$ cannot contain all three $v_1,v_2, v_3$.
\end{proof}

When none of the above simplifications are applicable, do the following:

\begin{enumerate}

\item Run the greedy algorithm to find an independent set $I_1$ in $H_1$.

\item In $H_2$ find a maximum independent set $I_2 \subset V_2$ with no edges to $I_1$, as in Lemma~\ref{lem:H2}.

\item Use Lemma~\ref{lem:transform} to output an independent set in $G$ of size $|I_1| + |I_2|$.

\end{enumerate}

The above lemmas imply that the algorithm can be run in polynomial time. We now analyse the size of the independent set that it finds.

When all simplification steps end, let $n_1$ and $cn_1$ denote the number of vertices and edge in $H_1$. Observe that $c \geq \frac{3}{2}$ because the minimum degree in $H_1$ is~$3$. Let $\alpha$ be the independence ratio of $H_1$, and observe that $\alpha \le 1/2$ (otherwise NT could have been applied). The number of vertices in $H_2$ is $n-n_1$. Since $G$ contains $n$ edges and our simplifications do not increase the number of edges, the number of edges in $H_2$ is at most $n-cn_1$. As every graph $G=(V,E)$ has an independent set of size at least $|V|-|E|$ it follows that the size of an independent set in $H_2$ is at least $(c-1)n_1$. Since the average degree in $H_1$ is $2c$ the greedy algorithm finds in $H_1$ an independent set of size $\frac{1 + \alpha^2}{2c + 1 + \alpha}n_1$. It follows that the approximation ratio of our algorithm is at least $\frac{c-1+\frac{1+\alpha^2}{2c+1+\alpha}}{c-1+\alpha}$. For $c \ge 3/2$ and $0 \le \alpha \le 1/2$ Lemma~\ref{lem:7/9} implies that the aforementioned expression is at least $7/9$, proving Theorem~\ref{thm:7/9}.
\end{proof}

\begin{lemma}
\label{lem:7/9}
Let $(x,y) \in [3/2,\infty) \times [0,1/2]$. Then \newline $f(x,y)=\frac{2x^2+y^2-x+xy-y}{2x^+y^2-x+3xy-1} \ge \frac{7}{9}$.
\end{lemma}
\begin{proof}
$\frac{2x^2+y^2-x+xy-y}{2x^+y^2-x+3xy-1} =1+\frac{1-y-2xy}{2x^+y^2-x+3xy-1}$. Thus it suffice to prove that for $x$ and $y$ in the above ranges $\frac{1-y-2xy}{2x^+y^2-x+3xy-1} \ge -\frac{2}{9}$ or $4x^2+2y^2-2x-12xy-9y+7 \ge 0$. Let $g(x,y)=4x^2+2y^2-2x-12xy-9y+7$. Clearly $\frac{\partial g(x,y)}{\partial x}=8x-2-12y > 0$ for such $x,y$ and $\frac{\partial g(x,y)}{\partial y}=4y-2x-9 < 0$ for such $x,y$. Hence $g(x,y) \ge g(3/2,1/2)=0$. As desired.
\end{proof}

Replacing the algorithm of Theorem~\ref{thm:isolated} by the algorithm of Theorem~\ref{thm:7/9} in algorithm PLG we obtain:

\begin{theorem}
\label{thm:recoverable7/9}
A canonical recoverable value $\rho = 7/3$ is achievable for MIS.
\end{theorem}

\section{MIS in $k$-colored graphs}

In analyzing algorithm PLG we only used the fact that $G_3$ has small average degree. However, $G_3$ has other structural properties as well. It is 3-colorable, and furthermore, the $3$-coloring can be found efficiently (e.g., by coloring the vertices of $G_3$ inductively in the order in which they appear in the permutation that generated $G_3$). We use the term $k$-colored graph to denote a graph with a given $k$-coloring. The following proposition is known~\cite{Hochbaum}.

\begin{proposition}
\label{pro:positive}
MIS on $k$-colored graphs can be approximated within $2/k$.
\end{proposition}

\begin{proof}
We present two proofs.

\noindent 1) Consider all possible pairs of color classes. At least one such pair contains at least $2/k$ fraction of opt. Each such pair is a bipartite graph on which MIS can be solved exactly.

\noindent 2) Use an LP for MIS which by its half integrality and standard rounding techniques allows us to assume that opt is at most half the graph. Thereafter, the largest color class provides at least a $2/k$ approximation.
\end{proof}

One may hope that Proposition~\ref{pro:positive} can be improved, for example, by combining the two different proofs given to it. If so, this would give an approximation ratio better than $2/3$ for $G_3$, hence potentially replacing or even surpassing the $5/7$ bound that we used for $G_3$ based on Theorem~\ref{thm:isolated}.

We show that Proposition~\ref{pro:positive} is nearly tight, unless vertex cover can be approximated within a ratio better than 2. (In particular, this implies UGC-hardness, since hardness of unique games is a stronger assumption than inapproximability of VC beyond a ratio of 2, see \cite{KR}.)

Let $G$ be an $n$-vertex graph in which one wants to approximate VC within a ratio better than 2. As is well known (e.g., \cite{HR97}), this is the same as distinguishing for some $\epsilon > 0$ whether $\alpha(G)$, the size of MIS in $G$, satisfies $\alpha(G) \ge (1/2 - \epsilon)n)$ or $\alpha(G) \le \epsilon n$.

For $k>2$ construct from $G$ a $k$-colored graph $G'$ as follows. For every vertex $v \in G$, the graph $G'$ contains $k$ copies $v_1, \ldots v_k$. All vertices with the same index $i$ form an independent set (hence a color class). Between any two distinct color classes other than class $k$, place a bipartite graph mimicking the edge pattern of $G$, connecting vertex $v_i$ with vertex $u_j$  (where $k \not= j \not= i \not= k$) if $(v,u)$ (or $(u,v)$) is an edge in $G$. Each vertex $v_k$ of the $k$th class is connected only to its own copies $v_i$ in the other classes. Hence the bipartite graph between class $k$ and any other class is simply a perfect matching.

\begin{lemma}
\label{lem:negative} In the reduction above,
$\alpha(G') = n + (k-2)\alpha(G)$.
\end{lemma}

This gives a gap of $\frac{n + O(\epsilon kn)}{kn/2 - O(\epsilon kn)} = (1 + o(1))\frac{2}{k}$ between {\em no} and {\em yes} instances.

To prove Lemma~\ref{lem:negative}, we shall use the following lemma.

\begin{lemma}
\label{lem:CopiesOfOneVertex}
There is a MIS $I'$ in $G'$ such that for every vertex $v \in G$, either $v_i \in I'$ for all $i \not= k$, or $v_k \in I'$.
\end{lemma}

\begin{proof}
Consider an arbitrary independent set $I'$ in $G'$. If either no copy or only one copy of $v$ is in $I'$, then without loss of size of $I'$ we may take this to be $v_k$. If at least two copies of $v$ are in $I'$, then neither one of them can be $v_k$. But then, we can add all other $v_i$ with $i \not= k$ to $I'$, since none of them can be a neighbor of a vertex already in $I'$.
\end{proof}

Lemma~\ref{lem:CopiesOfOneVertex} implies that all vertices $v \in G$ for which $v_k \not\in I'$ form an independent set in $G$. Lemma~\ref{lem:negative} easily follows.
Applying ~\ref{lem:negative} and \cite{KR} we immediately obtain:
\begin{theorem}
\label{thm:hard}
Let $k$ be an integer greater than $2$. Assume that for every $\epsilon > 0$ it is NP-hard to distinguish between graphs with independent set of size at least $(\frac{1}{2}-\epsilon)n$ to graphs with independent sets of size at most $\epsilon n$. Then, for arbitrary $\delta > 0$ it is NP-hard to approximate MIS on $k$-colored graphs within a factor larger than $\frac{2}{k}+\delta$. In particular, approximating MIS on $k$-colored graphs within a factor larger than $\frac{2}{k}+\delta$ is unique-games hard.
\end{theorem}

Recall that a vertex cover is the complement of an independent set. For $k$-colored graphs, vertex cover can be approximated within $2-2/k$. Theorem~\ref{thm:hard} implies (under the assumptions of the theorem) that this is best possible, up to low order terms.

Note that $G'$ constructed in the proof of Theorem~\ref{thm:hard} has additional structural properties beyond being $k$-colorable (many symmetries and an especially simple pattern of connections with color class $k$). This excludes better than $2/k$ approximation ratios even on more restricted families of graphs than those with a $k$-coloring. One such family is that of graphs composed of one part that is a tree, another that is an independent set, and edges between these two parts. These graphs are 3-colorable, but in addition, the pattern of edges between two of the color classes is very simple. Our proof implies that approximating MIS within a ratio better than 2/3 is hard on these graphs. Such graphs actually arise in a variation of the algorithm PLG, if instead of the first three layers one takes the first two layers, forming the part that is a tree, and the first layer in the {\em reverse} order of the permutation, forming the part that is an independent set. Note however that the graph $G_k$ used in algorithm PLG is $k$-colorable due to a structural property of being $(k-1)$-degenerate, whereas the graph $G'$ constructed in the proof of Theorem~\ref{thm:hard} is not necessarily $(k-1)$-degenerate. MIS is APX-hard even on $2$-degenerate graphs (this can be proved by removing one vertex from a $3$-regular graphs), but it appears possible to obtain algorithms that approximate MIS on such graphs within a ratio better than $2/3$. This may potentially lead to higher recoverable values than those shown in this work.

\section{Discussion}

We presented algorithms that achieve a recoverable value with $\rho = 2$ for MWIS and $\rho = 7/3 > 2$ for MIS. We also showed that previous algorithms (such as greedy and greedy+LP) fail to achieve a recoverable value with $\rho$ bounded away from~1. Finally, we proved a hardness of approximation result for MIS on graphs with a given $k$-coloring. Some questions remain open, most notably, how much higher can $\rho$ be pushed (recall that it cannot reach~4), and whether there truly is a gap between MIS and MWIS with respect to recoverable value.

Let us provide some comments on our permutation based algorithms. They are not guaranteed to produce a maximal independent set. After running them it may well be that there remain vertices in $G$ (rather than in $G_k$) not connected to the independent set that they find. This forms a residual graph. Hence if one wants to run our permutation based algorithms in practice, one can improve their performance by extracting an independent set also from the residual graph. We have not attempted to analyze the effect of this on the recoverable value.

The notion of recoverable value and the associated algorithms are most effective for graphs in which the vertices in the maximum independent set have small degree compared to the average degree in the graph. For some other classes of graphs (notably, regular graphs \cite{Berman,Chlebik}) our performance guarantees are poorer than those given by other algorithms.

\subsection*{Acknowledgements}

We thank Magnus Halldorsson for sharing with us his thoughts on the range of values of $d_{avg}$ for which Theorem~\ref{thm:isolated} holds.
{\small
\bibliographystyle{abbrv}
\bibliography{RVarXiv}
}

\begin{appendix}

\section{Some negative examples}

We show that  LP+greedy (mentioned in Section~\ref{sec:relatedwork}) does not meet the performance guarantees of Theorem~\ref{thm:MWIS}.

\begin{proposition}
\label{pro:negative}
For arbitrarily large integers $k,d$ there are graphs such that neither greedy nor LP+greedy find an independent set of size $k+1$ whereas the recoverable value is at least $\frac{\rho kd}{d+1}$. Hence they cannot guarantee a recoverable value with $\rho$ bounded away from~1.
\end{proposition}

\begin{proof}
Consider the following example. There are three layers. The first layer consists of $k$ vertices, each with $d$ distinct neighbors in the second layer. Every one of the $dk$ second layer vertices has $d-1$ distinct neighbors in the third layer. All $dk(d-1)$ vertices of the third layer form a clique. The second layer is an independent set with $dk$ vertices and
recoverable value of $dk\frac{\rho}{d+1}$. A greedy algorithm might
choose the $k$ vertices in the first layer and a clique vertex for a value of $k+1$.
This applies also to greedy+LP because the optimal half integral solution (to the standard LP, not the recoverable value LP)  of this instance will assign $1/2$ to all vertices, provided that $d$ and $k$ are large enough. (No third layer vertex will be assigned the value~1 since then all third layer vertices must be assigned~0. No second layer vertex will be assigned~1 because this forces $d-1$ distinct third layer vertices to be assigned~0 rather than~$1/2$. No first layer vertex will be assigned~1 because this forces $d$ distinct second layer vertices to be assigned~0 rather than~$1/2$.)
\end{proof}

We remark that the example in Proposition~\ref{pro:negative} (and also simpler examples suffice) shows that the fast algorithm of Theorem~\ref{thm:MWIS} does not offer an expected value of $\rho$ bounded away from~2. The expected number of vertices in the first two layers can be easily seen to be less than $2(k+1)$ (by using linearity of expectation to sum over individual vertices), and hence the expected size of the independent set found will not exceed $2(k+1)$ (which by increasing $d$ and $k$ can be made to imply $\rho$ arbitrarily close to~2).

The following graph shows that the RV LP together with our rounding procedure cannot give a value of $\rho > 2$ for MWIS. Consider a graph composed of clique of size $k$ connected via a complete bipartite graph to an independent set of size $k$. Vertices of the clique have weight $2k/(k+1)$ and other vertices have weight~1. The weight of maximum weight independent set is $k$ and its recoverable value is $\rho k/(k+1)$. The RV LP has an optimal solution that assigns value $1/2$ to all vertices. Thereafter, weighted greedy may choose a clique vertex in its first step, thus discarding all other vertices and finding a solution of weight $2k/(k+1)$. Hence $\rho = 2$.

\end{appendix}

\end{document}